\documentclass[submission,copyright,creativecommons]{eptcs}
\providecommand{\event}{AFL 2017} 
\usepackage{underscore}           
\usepackage{amssymb}
\usepackage{latexsym}
\usepackage{amsmath}
\usepackage{amsthm}
\usepackage{verbatim}
\usepackage{algorithm}
\usepackage[noend]{algorithmic}
\usepackage{graphicx}
\usepackage{xcolor}
\usepackage{epsfig}

\newcommand{\Z}{\mathbb{Z}}

\renewcommand{\d}{{\ensuremath{\diamond}}}

\newtheorem{proposition}{Proposition}
\newtheorem{theorem}{Theorem}

\newtheorem{conjecture}{Conjecture}

\title{Unavoidable Sets of Partial Words of\\ Uniform Length}
\author{Joey Becker
\institute{Department of Mathematics\\
University of Nebraska\\
P.O. Box 880130\\
Lincoln, NE 68588--0130, USA}
\and
F. Blanchet-Sadri 
\institute{Department of Computer Science\\
University of North Carolina\\
P.O. Box 26170\\
Greensboro, NC 27402--6170, USA}
\email{blanchet@uncg.edu}
\and
Laure Flapan 
\institute{Department of Mathematics\\
Yale University\\
P.O. Box 203164\\
New Haven, CT 06520, USA}
\and
Stephen Watkins 
\institute{Department of Mathematics\\
Vanderbilt University\\
PMB 352864\\
Nashville, TN 37235, USA}
}
\def\titlerunning{Unavoidable Sets of Partial Words of Uniform Length}
\def\authorrunning{J. Becker, F.~Blanchet-Sadri, L. Flapan \& S. Watkins}
\def\copyrightholders{J. Becker, F.~Blanchet-Sadri, L. Flapan \& S. Watkins}
\begin{document}
\maketitle

\begin{abstract}
A set $X$ of partial words over a finite alphabet $A$ is called unavoidable if every two-sided infinite word over $A$ has a factor compatible with an element of $X$. Unlike the case of a set of words without holes, the problem of deciding whether or not a given finite set of $n$ partial words over a $k$-letter alphabet is avoidable is NP-hard, even when we restrict to a set of partial words of uniform length. So classifying such sets, with parameters $k$ and $n$, as avoidable or unavoidable becomes an interesting problem. In this paper, we work towards this classification problem by investigating the maximum number of holes we can fill in unavoidable sets of partial words of uniform length over an alphabet of any fixed size, while maintaining the unavoidability property.


\end{abstract}

\section{Introduction}

The study of combinatorics on partial words has been developing in recent years (see, e.g., \cite{BSbook}). A partial word is a finite sequence over a finite alphabet $A$, a sequence that may have some undefined positions, called \emph{holes} and denoted by $\d$'s, where the $\d$ symbol is compatible with every letter of $A$. For example, $a{\d}{\d} c{\d} {\d} {\d} {\d} b$ is a partial word with six holes over the alphabet $\{a,b,c\}$. Now let $w$ be a two-sided infinite word and $u$ be a partial word.  Then, $w$ \emph{meets} $u$ if $w$ has a factor compatible with $u$; otherwise, $w$ \emph{avoids} $u$. A set $X$ of partial words over $A$ is {\em unavoidable} if every two-sided infinite word over $A$ meets an element of $X$; otherwise, it is {\em avoidable}. It is important to note that if $X$ is unavoidable, then every infinite unary word has a factor compatible with a member of $X$. Unavoidable sets of partial words were introduced in \cite{BSBrKaPaWe}. In the context of {\em total} words, those without holes, this concept of unavoidable sets has been extensively studied (see, e.g., \cite{Bel,ChHaPe,ChCu,CrLeWe,EvKi,Hig,HiSa,Ros95,Ros98,SaHi}).

There are two major problems that have been identified in the context of unavoidable sets of partial words. The first one is the problem of deciding whether a given finite set of partial words over a $k$-letter alphabet is avoidable, where $k\geq 2$. Unlike for total words, this problem is NP-hard \cite{BSJuPa} (see \cite{ChKa,Lot02} for an algorithm that efficiently decides the avoidability of sets of total words). While several variations of this problem are NP-hard, others are efficiently decidable \cite{BlBSGuRa,BSJiRe}. The second problem is the one of characterizing the unavoidable sets of $n$ partial words over an alphabet of size $k$. As shown in \cite{BSBrKaPaWe}, it is enough to consider the case where $k \leq n$ and when $k\geq 3$, the case where $k<n$. The $n=1$ and $k=1$ cases being trivial, the $n=2, k=2$ case was completely characterized by coloring Cayley graphs \cite{BSBlGuSiWebookchapter}. So the next step is to study the $n=3, k=2$ case. 

A problem, related to the characterization problem, we are concerned with is ``What is the minimum number of holes in an $m$-uniform unavoidable set of partial words (summed over all partial words in the set)?'' By $m$-uniform here, we mean each element in the set has constant length $m$. In \cite{BSChCh}, it was proved that for $m\geq 4$, the minimum number of holes in an $m$-uniform unavoidable set of size three over a binary alphabet is $2m-5$ if $m$ is even, and $2m-6$ if $m$ is odd. An easier way to think of it is the following.

\begin{theorem}\label{lastyear}\cite{BSChCh}
Let $m \geq 4$ and let $X = \{a{\d}^{m-2}a, b{\d}^{m-2}b, a{\d}^{m-2}b\}$ be an unavoidable set over $\{a, b\}$. Then the maximum number of holes we can fill in $X$, while maintaining the unavoidability property, is $m-1$ if $m$ is even, and $m$ if $m$ is odd. 
\end{theorem}

In this paper, given a $k$-letter alphabet $A_k = \{a_1, \ldots, a_k\}$, we consider subsets of $X_0 = \{a_i{\d}^{m-2}a_j \mid i \leq j\}$. We denote by $H^k_{m,n}$ the minimum number of holes in any unavoidable $m$-uniform set (summed over all partial words in the set) of size $n$ over $A_k$. Thus Theorem~\ref{lastyear} states that for $m\geq 4$, $H^2_{m,3}=2m-5$ if $m$ is even, and $H^2_{m,3}=2m-6$ if $m$ is odd. Without loss of generality, we require that $0, m-1$ are defined positions, i.e., $0, m-1$ are not holes, in each partial word in any unavoidable $m$-uniform set.

The contents of our paper are as follows. In Section~2, we review some background material on unavoidable sets of partial words. We also give the $k+{k\choose2}$ lower bound on the size of an $m$-uniform unavoidable set over $A_k$.  In Section~3, we give results on $m$-uniform unavoidable sets over $A_3$ which are useful to show our main result. In Section~4, we calculate the minimum number of holes in an $m$-uniform unavoidable set $X$ over $A_k$, where $X$ has size exactly $k+{k\choose2}$. In Section~5, we conclude with some remarks. 

\section{Preliminaries on unavoidable sets}

An {\em alphabet} $A$ is a non-empty finite set of {\em letters}. A \emph{finite word} over $A$ is a finite sequence of elements from $A$; in other words, it is a function $w: \{0,\ldots, |w|-1\} \to A$, where $|w|$ denotes the length of $w$. We write $w(i)$ for the letter at position $i$ of $w$ (positions are indexed starting at $0$).  

A \emph{two-sided infinite word} over $A$ is a function $w: \mathbb{Z} \to A$.  It is called $p$-periodic, or has period $p$, if $p$ is a positive integer such that $w(i)=w(i+p)$ for all $i\in \mathbb{Z}$.  For a non-empty finite word $v$, we write $v^\mathbb{Z}$ for the unique two-sided infinite $|v|$-periodic word $w$ such that $w(0)\cdots w(|v|-1)=v$, and we write $v^\mathbb{N}$ for the unique one-sided infinite $|v|$-periodic word $w$ such that $w(0)\cdots w(|v|-1)=v$. A finite word $u$ is a \emph{factor} of a two-sided infinite word $w$ if $w(i)\cdots w(i+|u|-1)=u$ for some $i\in \mathbb{Z}$.

A \emph{(finite) partial word} over $A$ is a function $u: \{0, \ldots, |u|-1\} \to A_\d$, where $A_\d=A\cup \{\d\}$ with $\d\not\in A$.  For $0\le i<|u|$, if $u(i)\in A$ then $i \in D(u)$ or $i$ is defined in $u$; otherwise, $i$ is a hole in $u$.  We write $h(u)$ for the number of holes in $u$. We say $u$ is a {\em total} word when $h(u)=0$. Letting $u$ and $v$ be two partial words of equal length, $u$ is \emph{compatible} with $v$, denoted $u\uparrow v$, if $u(i)=v(i)$ whenever $i\in D(u)\cap D(v)$. 

To {\em strengthen} a partial word is to replace a $\d$ with a letter in $A$, while to {\em weaken} a partial word is to set $u(i)=\d$ for some $i\in D(u)$. For example, $aa{\d}cb$ is a strengthening of $aa{\d}{\d} b$ and $a{\d}{\d}{\d} b$ is a weakening of $aa{\d}{\d} b$.  We say that we have ``filled a hole'' or ``inserted a letter'' in a partial word $u$ to mean that we have strengthened $u$.  We also say that the partial word $v$ is a {\em strengthening} of the partial word $u$, denoted $v\succ u$, if $v$ has a factor strengthening $u$. We similarly define {\em weakening}. 

We extend these notions to sets $X, Y$ of partial words as follows. The set $X$ is a strengthening of $Y$, denoted $X\succ Y$, if for every $x\in X$ there exists $y\in Y$ such that $x\succ y$.  Similarly for $X$ is a weakening of $Y$. It is important to note that if an infinite word $w$ meets a set $X$, then it also meets every weakening of $X$, while if $w$ avoids $X$ then it avoids any strengthening of $X$.  This means that if $X$ is unavoidable, so are all weakenings of $X$, while if $X$ is avoidable, so are all strengthenings of $X$.   

If $X$ is a set of partial words and $Y$ is the set resulting from performing operations on $X$ called factoring (if there exist partial words $x, y \in X$ such that $y$ is a weakening of a factor of $x$, then $Y=X\setminus \{x\}$), prefix-suffix (if there exists a partial word $x=ya\in X$ with $a\in A$ such that for every $b\in A$ there exists a suffix $z$ of $y$ and a partial word $v\in X$ with $v$ a weakening of $zb$, then $Y=(X\setminus \{x\})\cup \{y\}$), hole truncation (if $x{\d}^n \in X$ for some positive integer $n$, then $Y=(X\setminus \{x{\d}^n\})\cup \{x\}$), and expansion ($Y=(X\setminus \{x\})\cup \{x_1, x_2, \ldots, x_n\}$, where $\{x_1, x_2, \ldots, x_n\}$ is a partial expansion on $x\in X$), then $X$ is avoidable if and only if $Y$ is avoidable \cite{BSBrKaPaWe}. If $u = u_1{\diamond}u_2{\diamond} \ldots u_{n-1}{\diamond}u_n$, then $\{u_1 a_1 u_2 a_2 \ldots u_{n-1} a_{n-1}u_n \mid a_1, a_2, \ldots, a_{n-1} \in A\}$ is called a {\em partial expansion} on $u$ (note that $u_1, u_2, \ldots, u_n$ are partial words that may contain holes, and also note that $u$ is a weakening of $v$ for every member $v$ of a partial expansion on $u$). 

Letting $p$ be a prime and $q,m\in \mathbb{N}$, we write $p^q\|m$ if $p^q$ \emph{maximally divides} $m$, i.e., $p^q$ divides $m$, but $p^{q+1}$ does not divide $m$.

We end this section by establishing a lower bound on the size of an $m$-uniform unavoidable set over a $k$-ary alphabet. 

\begin{proposition}
\label{prop1}
There is no non-trivial unavoidable $m$-uniform set of size less than $k+{k\choose2}$ over $A_k$ (we call {\em trivial} any set of partial words containing the empty word or $\diamond^n$ for some positive integer $n$).
\end{proposition} 

\begin{proof}
Let $X$ be an $m$-uniform unavoidable set over $A_k$. None of the two-sided infinite words $w_i=a_i^\mathbb{Z}$ and $w_{i,j}=(a_i^{m-1}a_j^{m-1})^{\mathbb{Z}}$, $i < j$, can avoid $X$. Therefore, $X$ must contain an element compatible with a length $m$ factor of $w_i$ (by our convention, that element starts and ends with $a_i$), for each $i$, and an element compatible with a length $m$ factor of $w_{i,j}$ (by our convention, that element starts with $a_i$ and ends with $a_j$ or vice versa), for each $i<j$. Since these elements are distinct, we deduce that $|X| \geq k + {k \choose 2}$.
\end{proof}

\section{Uniform unavoidable sets over the ternary alphabet}

In examining the minimum number of holes in $m$-uniform unavoidable sets over $\{a,b,c\}$, we must consider sets of size at least $3 + {3 \choose 2} = 6$.  As mentioned earlier, we restrict our attention to sets of size exactly six. By the proof of Proposition~\ref{prop1}, $a {\d}^{m-2} a, b {\d}^{m-2} b, c {\d}^{m-2} c$ must be in the set, as well as one of $a{\d}^{m-2}b$ or $b{\d}^{m-2}a$, one of $a{\d}^{m-2}c$ or $c{\d}^{m-2}a$, and one of $b{\d}^{m-2}c$ or $c{\d}^{m-2}b$. There result eight possible sets:
\begin{center}
$\{a {\d}^{m-2} a, b {\d}^{m-2} b, c {\d}^{m-2} c, a {\d}^{m-2} b, a {\d}^{m-2} c, b {\d}^{m-2} c\}$,\\
$\{a {\d}^{m-2} a, b {\d}^{m-2} b, c {\d}^{m-2} c, a {\d}^{m-2} b, a {\d}^{m-2} c, c {\d}^{m-2} b\}$,\\
$\{a {\d}^{m-2} a, b {\d}^{m-2} b, c {\d}^{m-2} c, a {\d}^{m-2} b, c {\d}^{m-2} a, c {\d}^{m-2} b\}$,\\
$\{a {\d}^{m-2} a, b {\d}^{m-2} b, c {\d}^{m-2} c, b {\d}^{m-2} a, a {\d}^{m-2} c, b {\d}^{m-2} c\}$,\\
$\{a {\d}^{m-2} a, b {\d}^{m-2} b, c {\d}^{m-2} c, b {\d}^{m-2} a, c {\d}^{m-2} a, b {\d}^{m-2} c\}$,\\
$\{a {\d}^{m-2} a, b {\d}^{m-2} b, c {\d}^{m-2} c, b {\d}^{m-2} a, c {\d}^{m-2} a, c {\d}^{m-2} b\}$,\\
$\{a {\d}^{m-2} a, b {\d}^{m-2} b, c {\d}^{m-2} c, a {\d}^{m-2} b, c {\d}^{m-2} a, b {\d}^{m-2} c\}$,\\
$\{a {\d}^{m-2} a, b {\d}^{m-2} b, c {\d}^{m-2} c, b {\d}^{m-2} a, a {\d}^{m-2} c, c {\d}^{m-2} b\}$.\\
\end{center}
The last two are avoidable by $(a^{m-1}c^{m-1}b^{m-1})^{\mathbb{Z}}$ and $(b^{m-1}c^{m-1}a^{m-1})^{\mathbb{Z}}$ respectively, while the six others are equivalent up to renamings of letters (in fact, there is an unavoidable $m$-uniform set of minimal size for any total order on the alphabet). So we define the basic $m$-uniform unavoidable set of minimal size over $\{a,b,c\}$ as
$X_0 ={T_0}\cup{T'_0}$, where \begin{center} $T_0 = \{a {\d}^{m-2} a, b {\d}^{m-2} b, c {\d}^{m-2} c\}$ and $T'_0 = \{a {\d}^{m-2} b, a {\d}^{m-2} c, b {\d}^{m-2} c\}$.
\end{center}
The set $T_0$ contains only the words whose endpoints are the same, while $T'_0$ contains only those whose endpoints are different.  We begin by filling in the holes in $X_0$ one at a time to classify which strengthenings preserve unavoidability.  In the rest of the paper, the notation $X_i$ refers to a set created by filling in $i$ holes in $X_0$. 

\subsection{Filling in holes in $T_0$}

When we attempt to strengthen one of the $T_0$ words, we need to consider the following two cases: we can either insert $a$ into $a {\d}^{m-2} a$, say, to obtain $a{\d}^{x_1}a{\d}^{x_2}a$, or we can insert $b$, say, into $a {\d}^{m-2} a$ to obtain $a{\d}^{x_1}b{\d}^{x_2}a$.  However, filling in one hole of  $a {\d}^{m-2} a$ with an $a$ is equivalent to filling in any number of holes in $a {\d}^{m-2} a$ with $a$ if that is the only word we strengthen.

\begin{proposition}
\label{topsamedistinct} 
\begin{enumerate}
\item
For all $q\in \mathbb{N}$, the $m$-uniform set $X_q=(X_0\setminus\{a{\d}^{m-2}a\})\cup \{a{\d}^{x_1}a{\d}^{x_2}a\cdots a{\d}^{x_q}a\}$ is unavoidable.
\item
The $m$-uniform set $X_1=(X_0\setminus\{a {\d}^{m-2}a\})\cup \{a{\d}^{x_1}b{\d}^{x_2}a\}$ is avoidable. 
\end{enumerate}
\end{proposition}

Now, let us fill in two holes in two words of $T_0$. Recall that strengthenings of avoidable sets are avoidable.  We just noticed that if we want to preserve unavoidability, we cannot fill any of the holes in $a {\d}^{m-2}a$ with $b$.  Furthermore, filling in any number of the holes in $a {\d}^{m-2} a$ with $a$'s preserves unavoidability. The only remaining way to fill in two of the holes in $T_0$ is when one hole from $a {\d}^{m-2}a$ is filled with an $a$ and one hole from $b {\d}^{m-2}b$ is filled with a $b$.

\begin{proposition} 
\label{prop3}
Let $X_2=(X_0\setminus \{a{\d}^{m-2}a, b{\d}^{m-2}b\})  \cup \{a{\d}^{x_1}a{\d}^{x_2}a, b{\d}^{y_1}b{\d}^{y_2}b\}$ be an $m$-uniform set.  Furthermore, let $2^r\|x_1+1$, $2^s\|y_1+1$, and $2^t\|m-1$.  Then, $X_2$ is avoidable if and only if $m$ is odd and $r=s<t$.
\end{proposition}

Next, we address what happens when we fill in one hole in each of the three words in $T_0$. From \cite[Lemma~3]{BSBrKaPaWe}, there exists a two-sided infinite word $w_i$ over $\{a,b,c\}$ with period $m-1$ that avoids $\{a{\d}^ia, b{\d}^ib, c{\d}^ic\}$ for every $i\le \lfloor{\frac{m-3}{2}}\rfloor$. However, since $w_i$ has period $m-1$, $w_i$ avoids the set $$Z={T'_0} \cup \{a {\d}^{i} a{\d}^{m-i-3}a,b {\d}^{i} b{\d}^{m-i-3}b, c {\d}^{i} c{\d}^{m-i-3}c\}.$$

\begin{proposition}\label{three} Any $m$-uniform set of the form $X_3=T_3 \cup T'_0$, where $T_3=\{a {\d}^{x_1} a{\d}^{x_2}a, b {\d}^{y_1} b{\d}^{y_2}b, c {\d}^{z_1} c{\d}^{z_2}c\}$,  is avoidable.
\end{proposition}

\subsection{Filling in holes in $T'_0$}

First, there are two cases to consider when filling a hole in $a {\d}^{m-2} b$: the added letter is distinct from both $a$ and $b$, or it is  one of $a$ or $b$. For the former case, the $m$-uniform set $(X_0 \setminus \left\{a {\d}^{m-2} b\right\}) \cup
\left\{a {\d}^{x_1} c {\d}^{x_2} b\right\}$ is avoidable by the infinite word $(a^{m-1}b^{m-1})^\mathbb{Z}$. For the latter case, the following proposition holds.

\begin{proposition}\label{algo1prop}
\begin{enumerate}
\item If $x_2+1 \not\equiv 0 \pmod{x_1+1}$, the $m$-uniform set $X_1 = (X_0 \setminus \left\{a {\d}^{m-2} c\right\}) \cup
\left\{a {\d}^{x_1} c {\d}^{x_2} c\right\}$ is avoidable. Otherwise, $x_2 \geq x_1$ and $w = (a^{x_1+1}b^{x_2+1}c^{x_1+1}a^{x_2+1}b^{x_1+1}c^{x_2+1})^\mathbb{Z}$ avoids $X_1$.
\item If $x_1+1 \not\equiv 0 \pmod{x_2+1}$, then the $m$-uniform set $X_1 = (X_0 \setminus \left\{a {\d}^{m-2} c\right\}) \cup
\left\{a {\d}^{x_1} a {\d}^{x_2} c\right\}$ is avoidable. Otherwise, $x_1 \geq x_2$ and $w = (a^{x_1+1}b^{x_2+1}c^{x_1+1})^\mathbb{Z}$ avoids $X_1$.
\end{enumerate}
Consequently, if $(X_0\setminus \{a{\d}^{m-2}c\})\cup\{x\}$, where $x\uparrow a{\d}^{m-2}c$, is unavoidable, then $x$ has no interior defined positions.  
\end{proposition}

Filling holes in the words $a {\d}^{m-2} b$ and $b {\d}^{m-2} c$ is not as simple as filling holes in $a {\d}^{m-2} c$ while maintaining unavoidability.  We know that  inserting a letter different from the endpoints of the word into which it was inserted makes the resulting set avoidable.  Thus we only consider the case when we insert a letter that is the same as one of the endpoints. If we
insert one letter into $a {\d}^{m-2} b$, then no $c$ can appear in an avoiding word $w$.  This is because $w$ must avoid $a {\d}^{m-2} c$, $b {\d}^{m-2} c$, and $c {\d}^{m-2} c$.  Thus if $w$ were to contain a $c$, there would be no possible letter for the position $m-1$ spaces before the $c$.   Likewise, if we insert one letter into $b {\d}^{m-2} c$, any word $w$ which
contains an $a$ must meet one of $a {\d}^{m-2} a$, $a {\d}^{m-2} b$, or $a {\d}^{m-2} c$.  So in both of these cases, we are
reduced to the use of a binary alphabet. 

\begin{proposition}\label{harder} 
If any of the following conditions 1--4 hold, then the $m$-uniform set $X_1$ is unavoidable if and only if $r\le s$:
\begin{enumerate}
\item
$X_1=(X_0\setminus \{a{\d}^{m-2}b\}) \cup \{a{\d}^{x_1}b{\d}^{x_2}b\}$, and $2^s\|m-1$ and $2^r\|x_1+1$.  

\item
$X_1=(X_0\setminus \{a{\d}^{m-2}b\}) \cup \{a{\d}^{x_1}a{\d}^{x_2}b\}$, and $2^s\|m-1$ and $2^r\|x_2+1$.  

\item 
$X_1=(X_0\setminus \{b{\d}^{m-2}c\}) \cup \{b{\d}^{x_1}b{\d}^{x_2}c\}$, and $2^s\|m-1$ and $2^r\|x_2+1$.  

\item
$X_1=(X_0\setminus \{b{\d}^{m-2}c\}) \cup \{b{\d}^{x_1}c{\d}^{x_2}c\}$, and $2^s\|m-1$ and $2^r\|x_1+1$.  
\end{enumerate}

\end{proposition}

Second, let us fill in two holes in $T'_0$. We have seen that inserting any letter into $a {\d}^{m-2} c$ causes $X_0$ to become avoidable, but that inserting a letter into only one of
$a {\d}^{m-2} b$ or $b {\d}^{m-2} c$ only sometimes causes $X_0$ to become avoidable.  Inserting a letter in $a {\d}^{m-2} b$ or $b {\d}^{m-2} c$ that is different from both endpoints makes $X_0$ avoidable. Thus, in examining what happens when we fill in two holes in $T'_0$ we have two cases to consider. The first is when we fill two of the holes in either $a{\d}^{m-2}b$ or $b{\d}^{m-2}c$ with letters that match the endpoints.  The second case to consider is we fill one hole from $a{\d}^{m-2}b$ and one hole from $b{\d}^{m-2}c$ with letters matching one of the endpoints of their respective partial words.  We consider the first case first, for which results from \cite{BSChCh} prove useful.  

 Let $X_2$ be the set created from $X_0$ by filling in two of the holes in the same word in $T'_0$.   If we have filled in two holes in $a{\d}^{m-2}b$, then, as before, any word avoiding $X_2$ must be a word over $\{a,b\}$.  Similarly, if we have filled in two holes in $b{\d}^{m-2}c$, any word avoiding $X_2$ must be a word over $\{b,c\}$.  Let $Y$ be the set created by removing all of the elements of $X_2$ that contain the letter that cannot be contained in $X_2$'s avoiding word.  In either case, $X_2$ has the same avoidability as $Y$, since any word avoiding $Y$ automatically avoids $X_2$ and vice versa. The avoidability of $Y$ is completely characterized in \cite{BSChCh}.

\begin{theorem}\label{2bottom}\cite{BSChCh} Let $Y=\{a{\d}^{m-2}a, b{\d}^{m-2}b, a{\d}^{x_1}b{\d}^{x_2}b{\d}^{x_3}b\}$ be an $m$-uniform set over $\{a,b\}$.  Let $2^s\|m-1$,  $2^t\|x_1+1$, $2^r\|x_1+x_2+2$.  Then $Y$ is unavoidable if and only if $s\ge t,r$ holds in addition to one of $(i)$ $x_1=x_2$, $(ii)$ $x_1=x_3$, or $(iii)$ $m=7(x_1+1)+1$ and $x_2+1\in \{2(x_1+1), 4(x_1+1)\}$.
\end{theorem}

\begin{theorem}\label{lastyearswitcharound}\cite{BSChCh}
Let $i_1<\cdots<i_s<j_1<\cdots<j_r$ be elements of the set $\{1,\ldots,m-2\}$.  Let $x$ be defined as follows: $x(i)=a$ if $i\in\{0,i_1,\ldots,i_s\}$, $x(i)=b$ if $i\in\{j_1,\ldots,j_r,m-1\}$, and $x(i)={\d}$ otherwise. Then $Y= \{a{\d}^{m-2}a,b{\d}^{m-2}b, x\}$ has the same avoidability as some set $Z=\{a{\d}^{m-2}a,b{\d}^{m-2}b, z\}$, where $z$ is created by filling in $r+s$ of the holes in $a{\d}^{m-2}b$ with $b$'s.
\end{theorem}

We now focus on the set created by filling in one hole in $a{\d}^{m-2}b$ and one hole in $b{\d}^{m-2}c$.  We define such a set, with  $x_1 + x_2 = y_1 + y_2 = m - 3$, as \begin{equation} 
X_2 = T_0 \cup \{a {\d}^{x_1} b {\d}^{x_2} b, b {\d}^{y_1} b {\d}^{y_2} c, a {\d}^{m-2} c\}.\label{conjform}
\end{equation}
Such set has the same avoidability as $$Y_2= T_0 \cup \{a{\d}^{y_2}b{\d}^{y_1}b, b{\d}^{x_2}b{\d}^{x_1}c, a {\d}^{m-2} c\}.$$ 

\begin{proposition}\label{x1y2}  The $m$-uniform sets $$Y_2=(X_0\setminus \{a{\d}^{m-2}b,b{\d}^{m-2}c\}) \cup \{a{\d}^{y_2}b{\d}^{y_1}b,b{\d}^{x_2}b{\d}^{x_1}c\}$$ and $X_2$, defined by Eq.~(\ref{conjform}), have the same avoidability.
\end{proposition}

From Proposition~\ref{x1y2}, when considering $X_2$, defined by Eq.~(\ref{conjform}), we can assume without loss of generality that $x_1\le y_2$. If $y_2 < x_1$, $X_2$ has the same avoidability as the set $Y_2$ obtained by switching $x_1$ and $y_2$.

\begin{proposition}\label{avoidingword}  If $x_1\le x_2$, there exist integers $p,q>0$, with $p+q=m-1$, such that the infinite word $w=(a^pc^qb^p)^{\mathbb{Z}}$ avoids $X_2$, defined by Eq.~(\ref{conjform}).
\end{proposition}

\begin{proof}
Set $v=a^pc^qb^p$. If $w(i)=a$, we know $w(i+m-1)=b$.  Thus in order for $w$ to avoid $X_2$, we need $p\le x_2+1$, to ensure that $w$ avoids $a{\d}^{x_1}b{\d}^{x_2}b$.  Since $x_1+x_2+2=p+q=m-1$, $p\le x_2+1$ implies $q\ge x_1+1$.  Additionally, if $w(i)=c$, we need $q\le p$ in order to ensure that $w(i+m-1)$ is an $a$ and not a $c$. Finally, if $w(i)=b$, then $w(i+m-1) \in \{a, c\}$.  In fact, $m-1$ spaces after the first $p-q$ $b$'s in $v$ is an $a$ and $m-1$ spaces after the last $q$ $b$'s in $v$ is a $c$.  Thus to ensure that $w$ avoids $b {\d}^{y_1} b {\d}^{y_2} c$, we need $q\le y_2+1$.  Consequently, $w$ avoids $X_2$ if $x_1\le y_2$, which we have already assumed, and if we can find $p,q$ such that $x_1+1\le q\le p\le x_2+1$.  This occurs when $x_1\le x_2$.
\end{proof}

Thus the set $X_2$, defined by Eq.~(\ref{conjform}), is always avoidable except possibly when $y_1\le x_2 \le x_1\le y_2$. Extensive computations yield the following conjecture.

\begin{conjecture}
Set $X_2$, defined by Eq.~(\ref{conjform}), is avoidable when $y_1\le x_2 \le x_1\le y_2$.
\end{conjecture}

We now discuss some results towards a proof of this conjecture. Table~\ref{tab1} gives specific examples of words that avoid sets defined by Eq.~(\ref{conjform}) under conditions on $m,x_1,$ and $y_1$. We prove only the third item in Table~\ref{tab1}, i.e, Proposition~\ref{tab1-3}, as the proofs of the other items are analogous. 

\begin{proposition}
\label{tab1-3}
The infinite word $w=((ab)^pa(bc)^q)^\Z$, where $p \geq 0, q>0$, avoids $X_2$, defined by Eq.~(\ref{conjform}), if and only if the following conditions hold:
\begin{enumerate}
\item $m \equiv 2 \pmod{2(p+q)+1}$;
\item $x_1 \equiv 2j-1 \pmod{2(p+q)+1}$ for some $j \in [0..q]$;
\item $y_1 \equiv 2k-1 \pmod{2(p+q)+1}$ for some $k \in [q..p+q+1]$.
\end{enumerate}
\end{proposition}

\begin{proof}
For the remainder of the proof, assume all congruences are modulo $2(p+q) +1$. Suppose $X_2$ satisfies the above conditions. Since $m \equiv 2 $ by Condition~1 and, thus, $m-1 \equiv 1 $,  $w(i) = a$ implies $w(i+m-1) = b$ since any letter after an $a$ is a $b$.  Similarly, $w(i) = b$ implies $w(i+m-1) \in\{a,c\}$, and $w(i) = c$ implies $w(i+m-1) \in \{a, b\}$. Thus $w$ avoids $a{\d}^{m-2}a$, $b{\d}^{m-2}b$, $c{\d}^{m-2}c$, and $a{\d}^{m-2}c$. 

Suppose $w(i)=a$. Consider $w(i+x_1+1)$. By Condition~2, $x_1 \equiv 2j-1$ for some $j \in [0..q]$, which implies $w(i+x_1+1)=w(i+2j)$. Since any letter an even distance at most $2q$ spaces ahead of an $a$ is in $\{a,c\}$, $w(i+x_1+1) \neq b$. Therefore, $w$ avoids $a\d^{x_1}b$, which implies that $w$ avoids $a\d^{x_1}b\d^{x_2}b$. Next, suppose $w(i)=b$. By Condition~3, $y_1 \equiv 2k-1$ for some $k \in [q..p+q+1]$. Equivalently, $2k-1+y_2+3 \equiv y_1+y_2+3 = m \equiv 2$. Thus, $y_2+1 \equiv 2r$ for some $r \in [0..p+1]$. Since any letter an even distance at most $2p+2$ spaces ahead of  a $b$ is in $\{a,b\}$, $w(i+y_2+1)=w(i+2r)\neq c$. Therefore, $w$ avoids $b\d^{y_2}c$ and, thus, $b\d^{y_1}b\d^{y_2}c$. Therefore, if Conditions~1--3 are satisfied, then $w$ avoids $X_2$. 
 
Now, suppose $w$ avoids $X_2$.  We show that $X_2$ satisfies Conditions~1--3. Suppose for a contradiction that Condition~1 does not hold. Then either $m \equiv 2r+1$ for some $r \in [0..p+q]$ or $m \equiv 2r $ for some $r \in [2..p+q]$. Suppose $m \equiv 2r+1 $ for some $r \in [0..p+q]$. Without loss of generality, suppose $w(i)=a$ begins a period of $w$ and, thus, $w(i-2)=b$. Consider $w(i-2+m-1)=w(i+2r-2)$. Since any letter an even distance after the first letter in the period is in $\{a,c\}$, $w$ meets either $a{\d}^{m-2}a$ or $a{\d}^{m-2}c$, a contradiction. Similarly, suppose $m \equiv 2r$ for some $r \in [2..p+q]$ and $w(i)=a$ once again begins a period of $w$. Consider $w(i+m-1)=w(i+2r)$. Since any letter an even distance after the first letter in the period is in $\{a,c\}$, $w$ meets either $a{\d}^{m-2}a$ or $a{\d}^{m-2}c$, again a contradiction. Thus, Condition~1 holds. 

Next, suppose for a contradiction that Condition~2 does not hold.  There are two cases to consider. The first is that $x_1 \equiv 2r$ for some $r \in [0..p+q]$. Suppose $w(i)=a$ begins a period of $w$.  Since any letter an odd distance after the first letter in the period is a $b$, $w(i+x_1+1)=w(i+2r+1)=b$. Furthermore, since $w$ avoids $X_2$, if $w(i)=a$, then $w(i+m-1)=b$.  Thus, $w$ meets $a\d^{x_1}b\d^{x_2}b$, which is a contradiction. The second case is that $x_1+1 \equiv 2r$ for some $r \in (q..p+q)$. Once again, let $w(i)=a$ begin a period of $w$, so $w(i-2q-1)=a$. Then $w(i-2q-1+x_1+1)=w(i+2r-2q-1)$. Since any letter an odd number of spaces after the first $a$ in the period is a $b$, this means $w(i-2q-1+x_1+1)=b$.  Since $w$ avoids $X_2$, $w(i-2q-1+m-1)=b$. Thus $w$ meets $a{\d}^{x_1}b{\d}^{x_2}b$, which is a contradiction. Thus, Condition~2 holds as well.  

Finally, suppose for a contradiction that Condition~3 does not hold. This means $y_1 \not\equiv 2k-1$ for any $k \in [q..p+q+1]$. By Condition~1, $m \equiv 2$ and since $y_1+y_2+1=m-2 \equiv 0$, Condition~3 not holding is equivalent to $y_2+1 \equiv 2r+1$ for some $r\in [0..p+q)$ or $y_2+1 \equiv 2r$ for some $r\in [p+2..p+q]$. Suppose $y_2+1 \equiv 2r+1$ for some $r \in [0..p+q)$.  Now, let $w(i)=c$ be the last letter in a period of $w$. Thus $w(i-(y_2+1))=b$ and $w(i-(m-1))=b$, since any letter an odd number of spaces before the last $c$ in the period is a $b$ and since $w$ avoids $a\d^{m-2}c$ and $c\d^{m-2}c$. This contradicts the assumption that $w$ avoids $b\d^{y_1}b\d^{y_2}c$. Similarly, suppose $y_2+1 \equiv 2r$ for some $r\in [p+2..p+q]$. Let $w(i)=a$ begin a period of $w$ and $w(i-2)=b$. Consider $w(i-2+y_2+1)=w(i+2r-2)$. Since any letter an odd distance at least $2p+2$ spaces after the first letter in the period is a $c$, $w(i-2+y_2+1)=c$. Since $w$ avoids $a{\d}^{m-2}c$ and $c{\d}^{m-2}c$, we have that $w(i-2+y_2+1-(m-1))=b$. This means $w$ meets $b{\d}^{y_1}b{\d}^{y_2}c$, which is a contradiction. Thus, Condition~3 holds. 
\end{proof}

\begin{table}\scriptsize
\caption{Necessary and sufficient conditions for $w$ to avoid sets defined by Eq.~(\ref{conjform}) when $y_1 \le x_2 \le x_1 \le y_2$}\label{tab1}
\begin{center}
\begin{tabular}[t]{l|l}
Avoiding word $w$ & Necessary and sufficient conditions\\ 
\hline
\hline
{$(a^pb^p)^\Z$} &  $m \equiv p+1 \pmod{2p}$\\
& $x_1 \equiv -1 \pmod{2p}$\\ \hline

{$(b^pc^p)^\Z$} & $m \equiv p+1 \pmod{2p}$\\
& $y_1 \equiv p-1 \pmod{2p}$\\ \hline
   
{$((ab)^pa(bc)^q)^\Z$} &  $m \equiv 2 \pmod{2(p+q)+1}$\\
$p \geq 0, q>0$ &  $x_1 \equiv 2j-1 \pmod{2(p+q)+1}$, $j \in [0..q]$\\
 & $y_1 \equiv 2k-1 \pmod{2(p+q)+1}$, $k \in [q..q+p+1]$\\ \hline

{$(ab((ab)^pa(bc)^q)^r)^\Z$} &  $m \equiv 2 \pmod{r(2p+2q+1)+2}$\\
$p \geq 0, q>0$ &  $x_1 \equiv (2p+2q+1)j+2k-1 \pmod{r(2p+2q+1)+2}$\\ 
 & $\hspace{24pt} j \in [0..r],\,k \in [1..r]$\\
 & $y_1 \equiv (2q+2r+1)s+2t+1 \pmod{r(2p+2q+1)+2}$\\ 
 & $\hspace{24pt} s \in [0..r),\,t \in [q..p+q)\cup\{0\}$\\ \hline

{$(((ab)^pa(cb)^q)^r(ab)^pa(cb)^{q-1})^\Z$} & $m \equiv 0 \pmod{(r+1)(2p+2q+1)-2}$\\
$p \geq 0, q>0$ & $x_1 \equiv (2p+2q+1)j+2k \pmod{(r+1)(2p+2q+1)-2}$\\
 $r\geq 0$& $\hspace{24pt} j \in [0..r], \, k \in [p..p+q)$\\
  &  $y_1 \equiv (2p+2q+1)s+2t \pmod{(r+1)(2p+2q+1)-2}$ \\
 & $\hspace{24pt} s \in [0..r], \, t \in [-1..p]$\\ \hline

{$(a^pb^qc^r)^\Z$} & $m \equiv q+1 \pmod{p+q+r}$\\
$1\le p \le q$,  $1 \le r \le q$,& $x_1 \equiv \{p+q-1, \ldots, p+q+r-1\} \pmod{p+q+r}$\\
$q \le p+r$& $y_1 \equiv \{q-1, \ldots, p+q-1\} \pmod{p+q+r}$\\ \hline

{$(a^pc^rb^q)^\Z$} & $m \equiv p+r+1 \pmod{p+q+r}$\\
 $1\le p \le q$,  $1 \le r \le q$, &  $x_1 \equiv \{-1, \ldots, r-1\} \pmod{p+q+r}$\\
 $q \le p+r$& $y_1 \equiv \{r-1, \ldots, p+r-1\} \pmod{p+q+r}$\\ \hline			

{$(a^{p+1}b^{q-1}c^{r}a^{p}b^{q}c^{r-1})^\Z$} & $m \equiv p+r+2q \pmod{2p+2q+2r-1}$\\
$0\le p<q$,  $1 \le r \le q$, & $x_1 \equiv \{p+q-1, \ldots, p+q+r-2, 2p+2q+r-1,$\\
 $q \le p+r$ & $\hspace{24pt} \ldots, 2p+2q+2r-2 \} \pmod{2p+2q+2r-1}$\\
& $y_1 \equiv \{q-1, \ldots, p+q-1, p+2q +r-2,$ \\ 
 & $\hspace{24pt} \ldots, 2p+2q+r-2 \pmod{2p+2q+2r-1}$ \\ \hline

{$(a^{p-1}c^{r+1}b^{q-1}a^{p}c^{r}b^{q})^\Z$} & $m \equiv p+r+1 \pmod{2p+2q+2r-1}$\\
$1\leq p \leq q$,  $0 \le r < q$, & $x_1 \equiv \{1, \ldots,r-1, p+r+q-1,\ldots, p+2r+q-1 \}$\\
 $q \le p+r$ & $\hspace{18pt}  \pmod{2p+2q+2r-1}$\\
& $y_1 \equiv \{r-1, \ldots, p+r-1, 2p+r,\ldots, 2p+2r+q-2 \}$ \\ 
 & $\hspace{18pt}\pmod{2p+2q+2r-1}$ \\ \hline

{$(a^r(b^qc^q)^p)^{\Z}$} &  $m \equiv q+1 \pmod{2pq+r}$\\
$1 \leq r \leq q,\, p>0$ & $x_1 \equiv \{-1,2qj+k\} \pmod{2pq+r}$, \\
 & $\hspace{24pt} j\in [0..p), k \in [q+r-1..2q)$ \\
& $y_1 \equiv q-1 \pmod{2pq+r}$ \\ \hline

{$(a^r(c^qb^q)^p)^{\Z}$} &  $m \equiv -q+1 \pmod{2pq+r}$\\
$1 \leq r \leq q,\, p>0$ & $x_1 \equiv \{-1,2qj+k\} \pmod{2pq+r}$, \\
 & $\hspace{24pt} j\in [0..p), k \in [r-1..r+q-2]$ \\
& $y_1 \equiv \{-q-1, q-1\} \pmod{2pq+r}$ \\ \hline

{$(a^qb^qc^q(c^qb^{2q}c^q)^p)^{\Z}$} &  $m \equiv -2q+1 \pmod{(4p+3)q}$\\
$p \geq 0, q>0$ & $x_1 \equiv \{-1,4qj+k-1\} \pmod{(4p+3)q}$, \\
 & $\hspace{24pt} j\in [0..p], k \in [2q..3q]$ \\
& $y_1 \equiv \{-2q-1, 2q-1\} \pmod{(4p+3)q}$ \\ \hline

{$((a^pc^rb^q)^tb)^\Z$} &  $m \equiv q+2 \pmod{t(p+q+r)+1}$\\
$1 \leq p \leq q, $ &  $x_1 \equiv (2q+1)j+k-1 \pmod{t(p+q+r)+1}$\\ 
$ p+r=q+1$ & $\hspace{24pt} j \in [0..t],\,k \in [1..r]$\\
 & $y_1 \equiv (2q+1)h+i \pmod{t(p+q+r)+1}$\\ 
 & $\hspace{24pt} h \in [0..t],\,i \in [r..q]$\\ 

\end{tabular}
\end{center}
\end{table}
 
The following proposition also provides conditions for $X_2$ to be avoidable. 

\begin{proposition}\label{eveneven}
Let $X_2$ be as defined by Eq.~(\ref{conjform}). Then $X_2$ is avoided by an infinite word of period at most $m$ if one of the following conditions hold:
\begin{enumerate}
\item
$x_1, y_1$ are even and $y_1 \leq x_2 \leq x_1$;
\item
$y_1=0$ and $x_2 \leq x_1$.
\end{enumerate}  
\end{proposition}

Tables~\ref{tab21}, \ref{tab22}, and \ref{tab23} summarize some sufficient conditions for patterns to avoid $X_2$, defined by Eq.~(\ref{conjform}), with respect to residues modulo $2$, $3$, and $4$. 

\begin{table}
\caption{Sufficient conditions on residues modulo $2$ for $w$ to avoid sets defined by Eq.~(\ref{conjform})}\label{tab21}
\begin{center}
\begin{tabular}[t]{l|l|l|l}
Avoiding word $w$ & $m$ & $x_1$ & $y_1$ \\ 
\hline 
\hline
$(ab)^{\Z}$ & $0$ & $1$ & $0,1$ \\ 
\hline
$(bc)^{\Z}$ & $0$ & $0,1$ & $0$ \\ 
\hline
$((ab)^pa(cb)^q)^{\Z}$ & $1$ & $0$ & $0$   
\end{tabular}
\end{center}
\end{table}

\begin{table}
\caption{Sufficient conditions on residues modulo $3$ for $w$ to avoid sets defined by Eq.~(\ref{conjform})}\label{tab22}
\begin{center}
\begin{tabular}[t]{l|l|l|l}
Avoiding word $w$ & $m$ & $x_1$ & $y_1$ \\ 
\hline 
\hline
$(abc)^{\Z}$ & $2$ & $1,2$ & $0,1$ \\ \hline
$(acb)^{\Z}$ & $0$ & $0,2$ & $0,1$ \\ \hline
$(ab(abc)^p)^{\Z}$ & $1$ & $1$ & $0,=1$  \\ \hline
$((acb)^pb)^{\Z}$ & $1$ & $0$ & $=0,1$ 
\end{tabular}
\end{center}
\end{table}

\begin{table}
\caption{Sufficient conditions on residues modulo $4$ for $w$ to avoid sets defined by Eq.~(\ref{conjform})}\label{tab23}
\begin{center}
\begin{tabular}[t]{l|l|l|l}
Avoiding word $w$ & $m$ & $x_1$ & $y_1$ \\ 
\hline 
\hline
$(a^2b^2)^{\Z}$ & $3$ & $3$ & $0,1,2,3$ \\ \hline
$(b^2c^2)^{\Z}$ & $3$ & $0,1,2,3$ & $1$ 
\end{tabular}
\end{center}
\end{table}

If $X_2$, defined by Eq.~(\ref{conjform}), is avoidable, then all sets that contain strengthenings of two of the $T'_0$ words are avoidable. 

\begin{proposition}\label{switcharound}
Let $X_2$ be defined by Eq.~(\ref{conjform}) and let $X_2'=(X_2\setminus \{a{\d}^{x_1}b{\d}^{x_2}b\}) \cup \{a{\d}^{x_2}a{\d}^{x_1}b\}$. Also let $Y_2'=(X_2'\setminus \{a{\d}^{x_2}a{\d}^{x_1}b, b{\d}^{y_1}b{\d}^{y_2}c\}) \cup\{a{\d}^{y_2}b{\d}^{y_1}b,b{\d}^{x_1}c{\d}^{x_2}c\}$ and $Y_2=(Y_2'\setminus\{a{\d}^{y_2}a{\d}^{y_1}b\}) \cup \{a{\d}^{y_1}b{\d}^{y_2}b\}$ be $m$-uniform sets.
\begin{enumerate}
\item
If $X_2$ is avoidable, then $X_2'$ is avoidable.
\item
The sets $X_2'$ and $Y_2'$ have the same avoidability.
\item
If $Y_2'$ is avoidable, then $Y_2$ is avoidable. 
\end{enumerate}
\end{proposition}

Third, Theorems~\ref{2bottom} and \ref{lastyearswitcharound} state that filling in two holes in the same word in $T'_0$ only sometimes makes $X_0$ avoidable. We now prove that once we have filled in three holes in the same word in $T'_0$, $X_0$ becomes avoidable.  

\begin{proposition}\label{3bottom}
If the $m$-uniform set $X_3=(X_0\setminus\{a{\d}^{m-2}b\}) \cup \{x\}$ is unavoidable, where $x\uparrow a{\d}^{m-2}b$, then $x$ has at most two interior defined positions.
\end{proposition}

\begin{proof}
If more than two positions in $x$ have been filled, we know that they have to be filled with $a$'s or $b$'s otherwise $X_3$ would be avoidable.  However from \cite[Corollary~4]{BSChCh}, $\{a{\d}^{m-2}a$, $b{\d}^{m-2}b, x\}$ can be avoided by an infinite word $w$ over $\{a,b\}$.  This means that $w$ avoids $c{\d}^{m-2}c$, $a{\d}^{m-2}c$, and $b{\d}^{m-2}c$ as well.  Thus $w$ avoids all of $X_3$ and thus $X_3$ is avoidable.
\end{proof}

\subsection{Filling in holes in $T_0$ and $T'_0$}

Since filling in any of the holes in $a{\d}^{m-2}c$ results in an avoidable set, the only strengthenings of $T'_0$ we need to consider are strengthenings of $a{\d}^{m-2}b$ and $b{\d}^{m-2}c$.  Furthermore, in order to preserve unavoidability, we must fill in a word in $T_0$ with the same letter as its two endpoints.  Thus when filling in one hole in $T_0$ and one hole in $T'_0$, there are two possible cases to consider: the endpoints of the $T_0$ word are the same as one of the endpoints of the $T'_0$ word or the endpoints of the $T_0$ word are different from the two endpoints of the $T'_0$ word.  We now focus on the $m$-uniform set 
\begin{equation}
X_2=(X_0\setminus\{a{\d}^{m-2}a,b{\d}^{m-2}c\}) \cup \{a{\d}^{x_1}a{\d}^{x_2}a,b{\d}^{y_1}c{\d}^{y_2}c\}.\label{x2form}
\end{equation}
When considering it, we can assume without loss of generality that $x_1\le x_2$. Indeed, it is easy to show that the $m$-uniform set $X_2$, defined by Eq.~(\ref{x2form}), is avoidable if and only if the $m$-uniform set $X_2'=(X_0\setminus\{a{\d}^{m-2}a, b{\d}^{m-2}c\}) \cup\{a{\d}^{x_2}a{\d}^{x_1}a, b{\d}^{y_1}c{\d}^{y_2}c\}$ is avoidable.
It is also easy to show that if the $m$-uniform set $X_2$, defined by Eq.~(\ref{x2form}), is unavoidable, then so is $Y_2=(X_0\setminus\{a{\d}^{m-2}a,b{\d}^{m-2}c\}) \cup \{a{\d}^{x_1}a{\d}^{x_2}a,b{\d}^{y_2}b{\d}^{y_1}c\}$.

Table~\ref{tab3} gives some of the recurring patterns of words avoiding sets defined by Eq.~(\ref{x2form}). For instance, the last item in Table~\ref{tab3} translates as Proposition~\ref{prop20}.

\begin{table}
\caption{Conditions for $w$ to avoid sets defined by Eq.~(\ref{x2form}); here, $u_q$ denotes the $q$-length prefix of $u$,  $\overline{u_q}$ denotes the complement of $u_q$ where $\overline{b}=c$ and $\overline{c}=b$, and $p,q>0$ are integers such that $p+q=m-1$}\label{tab3}
\begin{center}
\begin{tabular}{l|l}
Conditions  &  Avoiding word $w$ \\ 
\hline
\hline
$m$ even, $y_1$ even & $ (a(bc)^{\frac{m-2}{2}}a(cb)^{\frac{m-2}{2}})^{\mathbb{Z}}$\\ \hline
$m$ odd, $x_1$ even, $y_1$ even & $((ab)^{\frac{m-1}{2}}(ac)^{\frac{m-1}{2}})^{\mathbb{Z}}$\\ \hline
$y_1\le x_1\le x_2\le y_2$ & $(a^pb^qa^pc^q)^{\mathbb{Z}}$ \\ \hline
$y_2\le x_1\le x_2\le y_1$ &  $(a^{p}u_{q}a^{p}\overline{u_{q}})^{\mathbb{Z}}$, where $u=(b^{y_2+1}c^{y_2+1})^{\mathbb{N}}$ \\
\end{tabular}
\end{center}
\end{table}

\begin{proposition}
\label{prop20}
Let $u=(b^{y_2+1}c^{y_2+1})^{\mathbb{N}}$. If $y_2\le x_1\le x_2\le y_1$, there exist integers $p,q>0$, $p+q=m-1$ such that the infinite word $w=v^{\mathbb{Z}}$ avoids $X_2$, defined by Eq.~(\ref{x2form}), where $v=a^{p}u_{q}a^{p}\overline{u_{q}}$ (here, $u_q$ denotes the $q$-length prefix of $u$ and $\overline{u_q}$ denotes the complement of $u_q$, where $\overline{b}=c$ and $\overline{c}=b$).
\end{proposition}

\begin{proposition}\label{iff} Let $X_2$, defined by Eq.~(\ref{x2form}), and $$Y_2'=(X_0\setminus\{a{\d}^{m-2}a, b{\d}^{m-2}c\}) \cup\{a{\d}^{x_1}a{\d}^{x_2}a, b{\d}^{y_2}b{\d}^{y_1}c\}$$ be $m$-uniform sets. Furthermore, let $2^s\|x_1+1$ and $2^t\|y_1+1$.  If $y_1=y_2$ and $s\not=t$, then $X_2$ and $Y_2'$ are unavoidable.
\end{proposition}

\begin{proof}
Suppose $y_1=y_2$ and $s\not=t$.  From now on, we refer to $y_1=y_2$ just as $y$.  By performing the operations of factoring, prefix-suffix, hole truncation, and expansion on $X_2$ from \cite{BSBrKaPaWe}, we obtain the set $$Y=\{a{\d}^{x_1}a, b{\d}^yb, b{\d}^yc, c{\d}^yb, c{\d}^yc, a{\d}^ya{\d}^yb, a{\d}^ya{\d}^yc, b{\d}^ya{\d}^yb, c{\d}^ya{\d}^yc\},$$ which has the same avoidability as $X_2$.

Assume for contradiction that $Y$ is avoidable. This implies there exists an infinite word $w$ that avoids $Y$. It is clear that $w$ cannot contain only $a$'s since $w$ must avoid $a{\d}^{x_1}a$.  Thus, $w$ must contain a $b$ or a $c$.  Without loss of generality let us assume that $w$ contains a $b$ since the argument if $w$ contains a $c$ is identical.  Now without loss of generality, assume $w(y+1)=b$.  Since $w$ avoids $b{\d}^yb$ and $b{\d}^yc$, this means $w(2(y+1))=a$.  Since $w$ avoids $b{\d}^ya{\d}^yb$, $w(3(y+1))\not=b$. If $w(3(y+1))=a$, $w(4(y+1))=a$ since $w$ avoids $a{\d}^ya{\d}^yb$ and $a{\d}^ya{\d}^yc$. But this means that $w(5(y+1))=a$ and so on.  Thus inductively, if $w(3(y+1))=a$, then  $w(p(y+1))=a$ for all $p \geq 2$. If $w(3(y+1))=c$, then $w(4(y+1))=a$ because $w$ must avoid $c{\d}^yb$ and $c{\d}^yc$. Since $w(4(y+1))=a$ and $w$ avoids $c{\d}^ya{\d}^yc$, then either $w$ can degenerate into a repeating string of $a$'s as before, or $w(5(y+1))=b$ and the sequence repeats.  Thus it is easy to see that $w$ must be made up of two possible strings of letters: 
\begin{center}
\item$a\underbrace{}_{y}a\underbrace{}_{y}a\underbrace{}_{y}a\underbrace{}_{y}a\underbrace{}_{y}a,$\\
\item$a\underbrace{}_{y}b\underbrace{}_{y}a\underbrace{}_{y}c\underbrace{}_{y}a\underbrace{}_{y}b.$
\end{center}
Thus $w$ must be $4(y+1)$-periodic. Since the period of $w$ must avoid $a{\d}^{x_1}a$, the period of $w$ cannot contain all $a$'s.  Thus the second string must occur in the period of $w$.

Without loss of generality assume $w(0)=a$, $w(y+1)=b$, $w(2(y+1))=a$, and $w(3(y+1))=c$.  This implies for $k\ge 1$ that $w(k(y+1)=a$ if $k$ is even and $w(k(y+1))\in \{b,c\}$ if $k$ is odd.  

Since $w(0)=a$ and $w$ avoids $a{\d}^{x_1}a$, $w(x_1+1)\in \{b,c\}$.  This means that $w(x_1+y+1)=a$.  Now assume $w(n(x_1+1) + n(y+1))=a$ and consider $w((n+1)(x_1+1) + (n+1)(y+1))$.  Since $w(n(x_1+1) + n(y+1))=a$, $w((n+1)(x_1+1) + n(y+1))\in \{b,c\}$ because $w$ avoids $a{\d}^{x_1}a$.  This means that $w((n+1)(x_1+1) + (n+1)(y+1))=a$.  So by induction,  $w(n(x_1+1) + n(y+1))=a$ for all $n\in \mathbb{N}$. 

Now consider $w(p(x_1+1) + q(y+1))$ for $p,q\in \mathbb{N}$ with one of $p,q$ even and the other odd. We know $p\pm r =q$ for some odd $r\in \mathbb{N}$. Thus, $w(p(x_1+1) + q(y+1))=w(p(x_1+1) + p(y+1) \pm r(y+1))\in \{b,c\}$.

Similarly, if we consider $w(p(x_1+1) + q(y+1))$ for $p,q\in \mathbb{N}$ with both of $p,q$ even or both of $p,q$ odd, $p\pm r =q$ for some even $r\in \mathbb{N}$.  Thus, $w(p(x_1+1) + q(y+1))=w(p(x_1+1) + p(y+1) \pm r(y+1))=a$.

Now, let $l$ be the least common multiple of $x_1+1$ and $y+1$.  Since $s\not=t$ the power of two that maximally divides $l$ is the same as the power of two that maximally divides one of $x_1+1$ and $y+1$ and is greater than the power of two that maximally divides the other.  Thus $l$ is even and $l=\alpha (x_1+1)$ and $l=\beta(y+1)$ where one of $\alpha,\beta$ is odd and the other is even.  This implies $w(\alpha(x_1+1)+\beta(y+1))\in \{b,c\}$.  However, $w(\alpha(x_1+1)+\beta(y+1))=w(2l)=w(2\beta(y+1))=a$, which is a contradiction.

Thus, $Y$ is unavoidable and so is $X_2$. The set $Y_2'$ is then unavoidable. To see this, assume for contradiction that there exists an infinite word $w$ that avoids $Y_2'$.  Since $X_2$ is unavoidable, $w$ must meet an element of $X_2$.  This means $w$ meets $b{\d}^{y_1}c{\d}^{y_2}c$.  Suppose $w(i)=b$,$w(i+y_1+1)=c$, and $w(i+y_1+1+y_2+1)=c$.  Since $w$ avoids $a{\d}^{m-2}c$ and $c{\d}^{m-2}c$, this means $w(i+y_1+1-(m-1))=w(i-(y_2+1))=b$.  Thus, $w(i-(y_2+1))=b$, $w(i)=b$, and $w(i+y_1+1)=c$. This contradicts the fact that $w$ avoids $b{\d}^{y_2}b{\d}^{y_1}c$.
\end{proof}

\section{Minimum number of holes in uniform unavoidable sets}

We now consider the minimum number of holes in an $m$-uniform unavoidable set of size $k+{k\choose2}$ over $A_k$.  To do this, our results from Section~3 prove useful. As discussed in Section~3, there is an unavoidable $m$-uniform set of minimal size for any total order on the alphabet and these sets are equivalent up to renamings of letters. So we define the basic $m$-uniform unavoidable set of minimal size over $A_k$ as $X_0 = {T_0}\cup {T'_0}$, where $T_0 = \{a_i {\d}^{m-2} a_i \mid 1 \leq i \leq k\}$ and $T'_0 = \{a_i {\d}^{m-2} a_j \mid 1 \leq i < j \leq k\}$.    

\begin{proposition}\label{bottomfour}
 Let $X_2=(X_0\setminus\{a_{i_1}{\d}^{m-2} a_{i_2}, a_{i_3} {\d}^{m-2} a_{i_4}\})\cup\{x,y\}$ where the integers $i_1,i_2,i_3,i_4$ are all distinct and where $x\uparrow a_{i_1}{\d}^{m-2}a_{i_2}$ and $y\uparrow a_{i_3}{\d}^{m-2}a_{i_4}$.  If $x$ and $y$ both have at least one defined interior position, then $X_2$ is avoidable. 
 \end{proposition}
 \begin{proof} 
 If we fill in $a_{i_1}{\d}^{m-2} a_{i_2}$ or $a_{i_3} {\d}^{m-2} a_{i_4}$ with letters different from their endpoints, we know that $X_2$ is avoidable by an infinite word over a ternary alphabet.  Thus, we must fill in  $a_{i_1}{\d}^{m-2} a_{i_2}$ and $a_{i_3} {\d}^{m-2} a_{i_4}$ with letters that are the same as their respective endpoints.  For ease of notation, we let $a_{i_1}=a, a_{i_2}=b, a_{i_3}=c, a_{i_4}=d$. Without loss of generality, assume $x=a{\d}^{x_1}b{\d}^{x_2}b$ and $y=c{\d}^{y_1}d{\d}^{y_2}d$.  Filling in more holes in $x$ and $y$ is just a strengthening of $X_2$.  Furthermore, filling in $a{\d}^{m-2}b$ with an $a$ instead of a $b$ or $c{\d}^{m-2}d$ with a $c$ instead of a $d$ yield an equivalent proof. 
We thus have eight cases:
 \begin{eqnarray}
 &x_1\le y_1\le y_2\le x_2;\label{1}&
\\ 
&x_1\le y_2\le y_1\le x_2;\label{2}&
\\
& y_1\le x_1\le x_2\le y_2;\label{3}&
\\
& y_2\le x_1\le x_2\le y_1;\label{4}&
\\
& x_2\le y_1\le y_2\le x_1;\label{5}&
\\
& x_2\le y_2\le y_1\le x_1;\label{6}&
\\
& y_1\le x_2\le x_1\le y_2;\label{7}&
\\
& y_2\le x_2\le x_1\le y_1.\label{8}&
 \end{eqnarray}
 
  In any infinite word $w$ that avoids $X_2$, if $w(i)=a$, $w(i+m-1)=b$ and $w(i+2(m-1))=a$ and similarly if $w(i)=c$, $w(i+m-1)=d$ and $w(i+2(m-1))=c$. So let $\overline{a}=b$, $\overline{b}=a$, $\overline{c}=d$, and $\overline{d}=c$.  Furthermore, given a one-sided infinite word $v$, let $v_i$ denote the prefix of $v$ of length $i$. Now, let $v=(a^{x_2+1}b^{x_2+1})^{\mathbb{N}}$ and $u=(c^{y_2+1}d^{y_2+1})^{\mathbb{N}}$.  Define the infinite word $w=(v_{p}u_{q}\overline{v_p}\overline{u_q})^{\mathbb{Z}}$ where $p,q>0$ and $p+q=m-1$.  The word $w$ avoids $X_2$ as long as $q>x_2$ and $p>y_2$.  This is because in $w$, $m-1$ spaces after every $a$ is a $b$ and $m-1$ spaces after every $b$ is an $a$ and similarly for $c$ and $d$.  Furthermore, as long as $q>x_2$ and $p>y_2$, if $w(i)=b$, then $w(i-(m-1))=a$ and if $w(i)=d$, then $w(i-(m-1))=c$.  Since $p+q=x_1+x_2+2=y_1+y_2+2=m-1$, $w$ avoids $X_2$ in Cases (\ref{4}), (\ref{5}), (\ref{6}), and (\ref{8}).
  
Let us now define the infinite word $w'=(a^pc^qb^pd^q)^{\mathbb{Z}}$ for some $p,q>0$ such that $p+q=m-1$.  We claim that $w'$ avoids $X_2$ as long as $p\le x_2+1$ and $q\le y_2+1$.  If $w'(i)=a$, then $w'(i+m-1)=b$, if $w'(i)=b$, then $w'(i+m-1)=a$, and similarly for $c$ and $d$.  Furthermore, as long as $p\le x_2+1$, $w'$ avoids $a{\d}^{x_1}b{\d}^{x_2}b$ and as long as $q\le y_2+1$, $w'$ avoids $c{\d}^{y_1}d{\d}^{y_2}d$.  Since $p+q=x_1+x_2+2=y_1+y_2+2=m-1$, $w'$ avoids $X_2$ in Cases (\ref{1}), (\ref{2}), (\ref{3}), and (\ref{7}).  

We have thus found infinite words that avoid $X_2$ for all eight cases and so $X_2$ is avoidable.
\end{proof}

\begin{proposition}\label{bottomendpoints}
Let $X_1=(X_0\setminus\{a_i{\d}^{m-2}a_{i+p}\}) \cup \{x\}$ where $k \ge i+p\ge i+2$, $x\uparrow a_i{\d}^{m-2}a_{i+p}$, and $x$ has at least one defined interior position.  Then $X_1$ is avoidable.
\end{proposition}

\begin{proof}
Suppose an infinite word $w$ avoids $X_1$ and contains only the letters $a_i$, $a_{i+1}$, and $a_{i+p}$. If $w(j)=a_i$, then $w(j+m-1)=a_{i+p}$ since $w$ must avoid $a_i{\d}^{m-2}a_i$ and $a_i{\d}^{m-2}a_{i+1}$.  If $w(j)=a_{i+1}$, then $w(j+m-1)=a_i$ since $w$ must avoid $a_{i+1}{\d}^{m-2}a_{i+1}$ and $a_{i+1}{\d}^{m-2}a_{i+p}$.  Finally, if $w(j)=a_{i+p}$, then $w(j+m-1)=a_i$ or $w(j+m-1)=a_{i+1}$ since $w$ must avoid $a_{i+p}{\d}^{m-2}a_{i+p}$. Therefore, the conditions on $a_{i},a_{i+1}$, and $a_{i+p}$ are identical to the conditions on the letters $a,b$, and $c$ when we considered the avoidability over $\{a,b,c\}$ of $\{a{\d}^{m-2}a,b{\d}^{m-2}b,c{\d}^{m-2}c,a{\d}^{m-2}b, b{\d}^{m-2}c,x\}$, where $x\uparrow a{\d}^{m-2}c$ and $x$ contains only $a$'s and $c$'s. Thus, the proof that we can generate such an avoiding word is identical to the proof of Proposition~\ref{algo1prop}.
\end{proof}

To prove our main result, we show that $X_0$ becomes avoidable once we fill in more than $m-1$ holes if $m$ is even and $m$ holes if $m$ is odd. 

\begin{theorem}
\label{mainresult}
 For $m\ge 4$, if Conjecture 1 is true, then the maximum number of holes we can fill into an $m$-uniform unavoidable set of size $k+{k\choose2}$ over $A_k$ is $m-1$ if $m$ is even and $m$ if $m$ is odd.  In other words,  $H^k_{m,k+{k\choose2}}=(k+{k\choose2})(m-2)-(m-1)$ if $m$ is even, and $H^k_{m,k+{k\choose2}}=(k+{k\choose2})(m-2)-m$ if $m$ is odd.
 \end{theorem}

\begin{proof}
When we fill in holes in $T_0$, say we fill in a hole in $a_i{\d}^{m-2}a_i$, the letter we fill in must be $a_i$ or else the infinite word $a_{i}^{\mathbb{Z}}$ avoids $X_0$ (see Proposition~\ref{topsamedistinct}). Additionally, filling in holes in more than two words in $T_0$ makes $X_0$ avoidable.   This is because by Proposition \ref{three}, if we fill in holes in three words in $T_0$, there exists an infinite word $w$  that avoids $X_0$ and that contains three distinct letters. Since $w$ does not contain any of the letters that make up the other elements of $X_0$, $w$ avoids all of the elements of $X_0$ and thus $X_0$ is avoidable. Thus we can fill holes into at most two of the words in $T_0$.

Using Proposition~\ref{switcharound} we prove that if Conjecture~1 is true, then filling in holes in two $T'_0$ words that have an endpoint in common makes $X_0$ avoidable. To prove this, it is enough to consider the 3-letter alphabet $\{a,b,c\}$. Let $Z_2=(X_0\setminus \{a{\d}^{m-2}b, b{\d}^{m-2}c\}) \cup \{x,y\}$ where $x\uparrow a{\d}^{m-2}b$, $y\uparrow b{\d}^{m-2}c$, and $x$ and $y$ each have at least one defined interior position.  We show that if Conjecture~1 is true, then $Z_2$ is avoidable. Indeed, we know that if the defined interior letter in either $x$ or $y$ is different from the endpoints of its respective word, then $Z_2$ is avoidable.  Thus, 
\begin{center}
$\begin{array}{ccccc}
X_2 & = & (X_0\setminus \{a{\d}^{m-2}b, b{\d}^{m-2}c\}) &\cup& \{a{\d}^{x_1}b{\d}^{x_2}b, b{\d}^{y_1}b{\d}^{y_2}c\}, \\
X_2' &=& (X_0\setminus \{a{\d}^{m-2}b, b{\d}^{m-2}c\}) &\cup& \{a{\d}^{x_2}a{\d}^{x_1}b, b{\d}^{y_1}b{\d}^{y_2}c\}, \\
Y_2&=&(X_0\setminus \{a{\d}^{m-2}b, b{\d}^{m-2}c\}) &\cup& \{a{\d}^{y_1}b{\d}^{y_2}b, b{\d}^{x_1}c{\d}^{x_2}c\}, \\ 
Y_2'&=&(X_0\setminus \{a{\d}^{m-2}b, b{\d}^{m-2}c\}) &\cup& \{a{\d}^{y_2}a{\d}^{y_1}b, b{\d}^{x_1}c{\d}^{x_2}c\}
\end{array}$ 
\end{center}
represent the only remaining cases to consider. If Conjecture~1 is true, then $X_2$ is avoidable for all $x_1,x_2,y_1,y_2>0$.  However if $X_2$ is avoidable for all $x_1,x_2,y_1,y_2>0$, this implies $X_2'$ is avoidable for all $x_1,x_2,y_1,y_2>0$, which then implies $Y_2'$ is avoidable for all $x_1,x_2,y_1,y_2>0$, which implies $Y_2$ is avoidable for all $x_1,x_2,y_1,y_2>0$.  

Using Proposition \ref{bottomfour}, filling in holes in two $T'_0$ words whose endpoints are all distinct also makes $X_0$ avoidable. Thus we can fill in holes in at most one word in $T'_0$.  Furthermore, we know that the letter we fill in must be the same as one of the endpoints.  From Proposition \ref{bottomendpoints}, the word we fill in must be of the form $a_i{\d}^{m-2}a_{i+1}$ and from Proposition~\ref{3bottom}, we cannot fill in more than two holes in any word in $T'_0$.

Thus if we want to preserve the unavoidability of $X_0$, we can fill in holes in at most two of the $T_0$ words and one of the $T'_0$ words.  Therefore, filling in holes in $X_0$ is equivalent to filling in holes in subsets of $X_0$ of size three, where each subset contains two words from $T_0$ and one word from $T'_0$.  So given a word $u$ in $T'_0$, either none of the two $T_0$ words share endpoints with $u$,  both of the two $T_0$ words share endpoints with $u$, or one of the two $T_0$ words shares an endpoint with $u$. Without loss of generality, these subsets are of three possible forms: 
\begin{center}
$\begin{array}{rcl}
Q&=&\{a_i{\d}^{m-2}a_i, a_j{\d}^{m-2}a_j, a_l{\d}^{m-2}a_{l+1}\},\\
R&=&\{a_i{\d}^{m-2}a_i, a_{i+1}{\d}^{m-2}a_{i+1}, a_i{\d}^{m-2}a_{i+1}\},\\
S&=&\{ a_i{\d}^{m-2}a_i, a_j{\d}^{m-2}a_j, a_i{\d}^{m-2}a_{i+1}\}.
\end{array}$
\end{center} 

We first consider $Q$. Let $Z=(X_0\setminus Q) \cup\{y,z,a_l{\d}^{x_1}d{\d}^{x_2}a_{l+1}\}$ where $y\uparrow a_i{\d}^{m-2}a_i$, $z\uparrow a_j{\d}^{m-2}a_j$, and $d\in\{a_l,a_{l+1}\}$.  We now show that if $h(y)+h(z)=m-2$ then $Z$ is avoidable. So suppose $h(y)+h(z)=m-2$. If $h(y)=0$, then $(a_j^{m-2}a_l a_j^{m-2}a_{l+1})^{\mathbb{Z}}$ avoids $Z$, and similarly, if $h(z)=0$, then $(a_i^{m-2}a_l a_i^{m-2}a_{l+1})^{\mathbb{Z}}$ avoids $Z$. Thus, suppose $h(y), h(z)\ge 1$.  Let $h(y)=n-2$ and $h(z)=m-n$. If the $m-n$ holes in $z$ are not consecutive, then the $(m-1)$-periodic word $(a_j^{n-1}a_i^{m-n})^{\mathbb{Z}}$ avoids $Z$, while if the $m-n$ holes in $z$ are consecutive, then the $(m-1)$-periodic word $(a_j^{n-2}a_ia_ja_i^{m-n-1})^{\mathbb{Z}}$ avoids $Z$. Filling in a second hole in $a_l{\d}^{m-2}a_{l+1}$ for a total of $m$ holes filled is just a strengthening of $Z$ and thus is also avoidable.  Thus, filling in more than $m-1$ holes in $Q$ makes $X_0$ avoidable.

We now consider $S$.  Let $Y=(X_0\setminus S) \cup \{x,y,z\}$ where $x\uparrow a_i{\d}^{m-2}a_i$, $y\uparrow a_j{\d}^{m-2}a_j$, and $z\uparrow a_i{\d}^{m-2}a_{i+1}$.  We show that filling in more than $m-1$ holes in $S$ makes $Y$ avoidable (and thus $X_0$ avoidable). As discussed above, we can assume that $x$ contains only the letter ${a_{i}}$ and $y$ contains only the letter ${a_{j}}$. If $h(y)=0$, then filling in any of the holes in $S$ is equivalent to filling in holes in $R$, which we do below.  Thus, we assume $h(y)\ge 1$.   

Let $Y'=\{x, y, {a_{i}}{\d}^{ x_1}{a_{i}}{\d}^{x_2}{a_{i+1}},{a_{i+1}}{\d}^{m-2}{a_{i+1}},{a_{j}}{\d}^{m-2}{a_{i}},{a_{j}}{\d}^{m-2}{a_{i+1}}\}$. We now show that if $h(x)+h(y)=m-2$, then $Y'$ is avoidable. If $h(x)=0$, then $({a_{j}}^{m-2}{a_{i}}{a_{j}}^{m-2}{a_{i+1}})^{\mathbb{Z}}$ avoids $Y'$. Thus, assume $h(x)\ge 1$.  Let $h(x)=m-n$ and $h(y)=n-2$.  First, suppose the $m-n$ holes in $x$ do not appear in a contiguous block.  Then the $(m-1)$-periodic word $w=({a_{i}}^{n-1}{a_{j}}^{m-n})^{\mathbb{Z}}$ avoids $Y'$.  Since $w$ is $(m-1)$-periodic, it does not meet ${a_{j}}{\d}^{m-2}{a_{i}}$.  Since $w$ does not contain any ${a_{i+1}}$'s it avoids ${a_{i+1}}{\d}^{m-2}{a_{i+1}}$, ${a_{j}}{\d}^{m-2}{a_{i+1}}$, and ${a_{i}}{\d}^{ x_1}{a_{i}}{\d}^{x_2}{a_{i+1}}$. Let $u$ be an $m$-length factor of $w$ such that $u(0)={a_{j}}$.  We know that $u$ contains $n-1$ consecutive occurrences of ${a_{i}}$.  Since $h(y)=n-2$, there is at least one instance where $u$ has an ${a_{i}}$ in a position where $y$ has an ${a_{j}}$.  Thus, $w$ does not meet $y$.  Similarly, let $v$ be an $m$-length factor of $w$ such that $v(0)={a_{i}}$.  This means $v$ contains a contiguous block of $m-n$ ${a_{j}}$'s.  However,  $v\not\uparrow x$ since the holes in $x$ do not form a contiguous block. Now, suppose the $m-n$ holes in $x$ appear in a contiguous block. Then the $(m-1)$-periodic word $w'=({a_{i}}^{n-2}{a_{j}}{a_{i}}{a_{j}}^{m-n-1})^{\mathbb{Z}}$ avoids $Y'$.  It avoids ${a_{j}}{\d}^{m-2}{a_{i}}$, ${a_{j}}{\d}^{m-2}{a_{i+1}}$, ${a_{i}}{\d}^{ x_1}{a_{i}}{\d}^{x_2}{a_{i+1}}$, ${a_{i+1}}{\d}^{m-2}{a_{i+1}}$,  and $y$ for the same reasons that $w$ does.  However, since the $m-n$ holes in $x$ appear in a contiguous block, and there are $m-n$ ${a_{j}}$'s in $w'$ that are not situated in a contiguous block, $w'$ avoids $x$.
Thus we have shown that filling in $m-2$ holes in $T_0$ and filling in a hole with ${a_{i}}$ in ${a_{i}}{\d}^{m-2}{a_{i+1}}$ makes $Y'$ avoidable. If we fill in a second hole with ${a_{i}}$ in ${a_{i}}{\d}^{m-2}{a_{i+1}}$, for a total of $m$ holes filled, this is just a strengthening of the previous case and thus is also avoidable. Furthermore, by Theorem~\ref{lastyearswitcharound} substituting ${a_{i+1}}$'s for the ${a_{i}}$'s would yield the same avoidability.

We finally consider $R$.  Suppose an infinite word $w$ avoids $X=(X_0\setminus R) \cup \{x,y,z\}$ where $x\uparrow a_i{\d}^{m-2}a_i$, $y\uparrow a_{i+1}{\d}^{m-2}a_{i+1}$, and $z\uparrow a_i{\d}^{m-2}a_{i+1}$. Since we want to show that we can fill in $m-1$ holes, suppose at least two of $x,y,z$ have some defined interior positions. We prove that $w$ must be over the binary alphabet $\{a_i, a_{i+1}\}$ by considering two cases. First, suppose $a_i=a_1$ (the proof is similar if $a_{i+1}=a_k$). If $w(m-1)=a_k$, then no letter works for $w(0)$ since $w$ must avoid $a_j{\d}^{m-2}a_k$ for all $j\in \{1,\ldots,k\}$; thus, $w$ does not contain any $a_k$'s.  Similarly if $w(m-1)=a_{k-1}$, then no letter works for $w(0)$ since $w$ must avoid $a_j{\d}^{m-2}a_{k-1}$, for all $j\in \{1,\ldots,k-1\}$, and $w$ does not contain any $a_k$'s; thus, $w$ cannot contain any $a_{k-1}$'s.  We can continue eliminating potential letters from $w$ until we are left with only $a_1$ and $a_2$.  If $w(m-1)=a_2$, then $w(0) \in \{a_1, a_2\}$ depending on which of $x,y,z$ have defined interior positions.  Similarly,  if $w(m-1)=a_1$, then $w(0) \in \{a_1, a_2\}$.  Thus, $w$ is over $\{a_1, a_2\}$. Now, suppose $a_i \not = a_1$ and $a_{i+1} \not = a_k$. If $w(m-1)=a_k$, then as above we can show that $w$ cannot contain any of $a_{i+2},\ldots,a_k$, and if $w(0)=a_1$, that $w$ cannot contain any of $a_1,\ldots,a_{i-1}$.  If $w(0)=a_i$, then $w(m-1) \in \{a_i, a_{i+1}\}$ depending on which of $x,y,z$ have defined interior positions.  Similarly,  if $w(0)=a_{i+1}$, then $w(m-1) \in \{a_i, a_{i+1}\}$. 

We have shown that any infinite word that avoids $X$ must be over the alphabet $\{a_i,a_{i+1}\}$.  Thus, by Theorem~\ref{lastyear}, the maximum number of holes we can fill in $X$ while maintaining the unavoidability property is $m-1$  if $m$ is even and $m$ if $m$ is odd.  
\end{proof}

\section{Conclusion}

In this paper, we have considered $m$-uniform unavoidable sets of partial words over an arbitrary alphabet $A_k=\{a_1, \ldots, a_k\}$. We have formulated a conjecture, Conjecture~1, that states that the sets defined by Eq.~(\ref{conjform}) 
are avoidable when $y_1\le x_2\le x_1\le y_2$ and $a, b, c$ are distinct letters. If Conjecture~1 is true, for $m\geq 4$, we have exhibited a formula that calculates the maximum number of holes we can fill in any $m$-uniform unavoidable set of partial words over $A_k$, while maintaining the unavoidability property.

We believe that Conjecture~1 is true and have tested it for all $m$-uniform sets defined by Eq.~(\ref{conjform}) up to $m=100$ that satisfy $y_1\le x_2 \le x_1\le y_2$.  We have found that these sets are all avoidable.  In fact, all of the sets we tested have an avoiding word with period less than $2m$.  Of the $41,650$ such sets, only $4$ were found to require avoiding words that did not match any of our patterns. Furthermore, only $77$ of the roughly $42$ million sets for $m\leq1000$ are not covered by our patterns. However, we are doubtful that a small number of similar patterns could be shown to cover the remaining cases.

\section*{Acknowledgements}

This material is based upon work supported by the National Science Foundation under Grant Nos. DMS--0754154 and DMS--1060775. The Department of Defense is gratefully acknowledged. We thank Andrew Lohr as well as the referees of preliminary versions of this paper for their very valuable comments and suggestions. A research assignment from the University of North Carolina at Greensboro for the second author is also gratefully acknowledged. Some of this assignment was spent at the IRIF: Institut de Recherche en Informatique Fondamentale of Universit\'{e} Paris-Diderot--Paris 7, France.

\bibliographystyle{eptcs}

\providecommand{\urlalt}[2]{\href{#1}{#2}} 
\providecommand{\doi}[1]{doi:\urlalt{http://dx.doi.org/#1}{#1}}

\end{document}